\documentclass[a4paper]{article}
\usepackage{latexsym,amsmath,amssymb,enumerate,mathrsfs,amsthm}
\usepackage{graphicx,color}

\numberwithin{equation}{section}

\newcommand{\alg}[1]{\mathfrak{#1}}

\newcommand{\mc}[1]{\mathcal{#1}}

\newcommand{\Ad}{\operatorname{Ad}}

\newcommand{\bonds}{{\bf B}}

\newcommand{\wlim}{\operatorname{w-lim}}

\theoremstyle{plain}
\newtheorem{theorem}{Theorem}
\numberwithin{theorem}{section}
\newtheorem{lemma}[theorem]{Lemma}
\newtheorem{proposition}[theorem]{Proposition}

\theoremstyle{definition}
\newtheorem{definition}[theorem]{Definition}

\theoremstyle{remark}
\newtheorem{remark}[theorem]{Remark}

\title{Kosaki-Longo index and classification of charges in 2D quantum spin models}

\author{Pieter Naaijkens\footnote{E-mail: \texttt{pieter.naaijkens@itp.uni-hannover.de}}\\
Institute for Theoretical Physics\\Leibniz Universit\"at Hannover, Germany}
\begin{document}
\maketitle
\begin{abstract}
We consider charge superselection sectors of two-dimensional quantum spin models corresponding to cone localisable charges, and prove that the number of equivalence classes of such charges is bounded by the Kosaki-Longo index of an inclusion of certain observable algebras. To demonstrate the power of this result we apply the theory to the toric code on a 2D infinite lattice. For this model we can compute the index of this inclusion, and conclude that there are four distinct irreducible charges in this model, in accordance with the analysis of the toric code model on compact surfaces. We also give a sufficient criterion for the non-degeneracy of the charge sectors, in the sense that Verlinde's matrix $S$ is invertible.
\end{abstract}

\section{Introduction}
Over the last decade the interest in quantum systems which exhibit (quasi)par\-ticle excitations with anyonic or braided statistics has greatly increased. One of the driving forces behind this is the realisation that such systems may be useful for quantum information. In particular, topological features of such excitations are thought to be helpful in protecting the storage of quantum information from interactions with the environment, or to implement inherently fault-tolerant quantum circuits~\cite{MR1612425,MR1951039,MR2443722}. In recent work we have started to investigate a framework to obtain all relevant properties of such excitations from first principles~\cite{toricendo,phdnaaijkens}, where the focus is on the thermodynamic limit of quantum spin models. Essentially, the excitations can be described by certain representations of the observable algebra.

In concrete models it is often possible to explicitly construct such representations for different types of excitations (or charges). A natural question is then if one has found \emph{all} such excitations. In the present work we show that this question can be answered by looking at a certain numerical value, the \emph{Kosaki-Longo index}, associated to inclusions of certain von Neumann algebras generated by observables.

This index is related to the index that Jones introduced for inclusions $\alg{N} \subset \alg{M}$ of Type II$_1$ von Neumann algebras~\cite{MR696688}. Kosaki later extended Jones' notion of an index to inclusions of arbitrary factors~\cite{MR829381}. A different approach to define the index was taken by Longo, who also studied connections with the ``statistical dimension'' of superselection sectors in quantum field theory~\cite{MR1027496,MR1059320}. One can show that the Kosaki and Longo indices coincide, and we will henceforth call it the Kosaki-Longo index. A review of the subject can be found in~\cite{MR1662525}.

To recall the definition of this index, consider a subfactor $\alg{N} \subset \alg{M}$. We will always assume that $\alg{N}$ and $\alg{M}$ are infinite factors, and for our purposes we can restrict to irreducible inclusions, that is, inclusions for which $\alg{M} \cap \alg{N}' = \mathbb{C} I$. The subfactor $\alg{N} \subset \alg{M}$ has \emph{finite index} if there is some $\lambda > 0$ and a normal conditional expectation $\mc{E}$ of $\alg{M}$ onto $\alg{N}$ such that
\begin{equation}
	\label{eq:pimpop}
	\mc{E}(X) \geq \lambda X \quad \textrm{for all}\, X \in \alg{M}_{+}.
\end{equation}
Note that since $\lambda > 0$, such an $\mc{E}$ is automatically faithful. The index is then defined as $[ \alg{M} : \alg{N} ]_{\mc{E}} = \lambda^{-1}$, where $\lambda$ is the best such constant satisfying the inequality. In fact, for irreducible inclusions there is at most one normal conditional expectation~\cite[Sect. 3.3]{MR1662525}. For completeness we mention that in the non-irreducible case there is a unique conditional expectation $\mc{E}$ that minimizes $[\alg{M}:\alg{N}]_{\mc{E}}$ by~\cite[Thm. 5.5]{MR1027496} if the index is finite. We will denote the corresponding index by $[ \alg{M} : \alg{N}]$. If such a conditional expectation does not exist, we set $[ \alg{M} : \alg{N} ] = \infty$. Also note that normality of the conditional expectation already follows when it satisfies a Pimsner-Popa bound (equation~\eqref{eq:pimpop}) with $\lambda>0$ by~\cite[Cor. A.4]{MR2885606}.

Here we will be interested in the Kosaki-Longo index for certain inclusions of observable algebras. More concretely, we consider quantum spin models on infinite 2D lattices. Consider an irreducible ground state (with respect to some dynamics) representation $\pi_0$ of a $C^*$-algebra $\alg{A}$ of observables, satisfying some technical properties that we will explain later.\footnote{In particular Haag duality and the distal split property for cones.} Let $\Gamma$ be the union of two disjoint (and sufficiently far removed) cone-like regions. Then one can associate a von Neumann algebra $\mc{R}(\Gamma) := \pi_0(\alg{A}(\Gamma))''$ to this region, generated by all observables localised in $\Gamma$. Similarly one defines $\widehat{\mc{R}}(\Gamma) := \pi_0(\alg{A}(\Gamma^c))' \supset \mc{R}(\Gamma)$, where $\Gamma^c$ is the complement of $\Gamma$. To such inclusions of von Neumann algebras one can associate an index as above. The main goal of the present work is to show that this quantity tells us something about the number of elementary excitations (or charges) of our quantum spin model. Our methods are inspired by similar results that have been obtained in the context of (rational) conformal field theory on the circle~\cite{MR1838752}.

The essential connection between the Kosaki-Longo index and the number of distinct charges of the system is facilitated by identifying the different charges as (equivalence classes of) endomorphisms of the observable algebra. In essence, these endomorphisms describe how the observables in the ground state are modified in the presence of a \emph{single} excitation in the system. Such endomorphisms are particularly suited to topological models where excitations are created in pairs connected by a string. A prominent example is Kitaev's quantum double model~\cite{MR1951039}. The endomorphisms are then obtained essentially by sending one excitation of the pair to infinity, c.f.~\cite{toricendo}. Different types of excitations lead to inequivalent endomorphisms. The corresponding equivalence classes of representations obtained from these endomorphisms are also called \emph{superselection sectors}~\cite{MR1405610}. 

Another interesting and important question is what the \emph{statistics} of the excitations are. Or in other words, how do they behave under interchange? Two dimensional models are especially intersting in that respect, since there one can in principle have anyonic or braided statistics that go beyond the traditional Bose/Fermi (para)statistics~\cite{MR1016869,MR1081990}. The properties of the mutual statistics of the charges can be encoded in the ``modular matrix'' $S$~\cite{MR954762}. If this matrix $S$ is invertible, the behaviour is as far away as possible from the Bose/Fermi alternative. This can be interpreted as a certain non-degeneracy condition. Here we will give a sufficient criterion for $S$ to be invertible, stated only in terms of certain observable algebras. This generalises a result of M\"uger for quantum field theories on the line~\cite{MR1721563}. In contrast to the quantum field theory case, 2D lattice systems appear to be more amenable to a direct verification of this criterion. For example, the proofs in Sect.~\ref{sec:toric} of the present paper immediately imply that this condition is satisfied for the toric code.

The outline of the paper is as follows. First we will state the framework in which we will work in Section~\ref{sec:charges}. In Section~\ref{sec:index} the cone index associated to certain observable algebras is defined, and some properties are discussed. In Section~\ref{sec:toric} the theory will be applied to Kitaev's toric code model: in particular we show that we have a complete classification of the charges in that model. Finally, in the last Section we give a sufficient criterion for the non-degeneracy of the charges.

\section{Charges in 2D spin models}\label{sec:charges}
In this section we will outline the mathematical framework that we will use. For the sake of concreteness, and because of the application that we consider later, we assume that we have a $\mathbb{Z}^2$ lattice. Both directions in the lattice are seen as \emph{spatial} directions. The continuous time direction only plays a minor role in our analysis. The edges between neighbouring vertices of the lattice are called \emph{bonds}. The set of bonds will be denoted by $\bonds$. We assume that at each bond there is a quantum system with finitely many degrees of freedom, for example a spin-$d$ system. One can easily generalise this to more complicated situations, but this will suffice for us.

Note that we are dealing with systems with infinitely many sites. Such systems can conveniently be described in the $C^*$-algebraic approach to quantum systems~\cite{MR887100,MR1441540}. We recall the essential points here. At each bond $b$ the quantum system is described by a $d$-dimensional Hilbert space, with algebra of observable $\alg{A}(\{b\}) = M_d(\mathbb{C})$. The observable algebra of a finite set $\Lambda \in \mc{P}_f(\bonds)$ of sites, where $\mc{P}_f(\bonds)$ is the set of all finite subsets of $\bonds$, is obtained as usual: $\alg{A}(\Lambda) := \bigotimes_{j \in \Lambda} \alg{A}(\{j\})$. Hence it is isomorphic to the tensor product of $|\Lambda|$ copies of $M_d(\mathbb{C})$. 

The assignment $\Lambda \mapsto \alg{A}(\Lambda)$ of algebras to finite regions of the system has two important properties. If $\Lambda_1 \subset \Lambda_2$ are finite subsets of $\bonds$ then there is a natural inclusion $\alg{A}(\Lambda_1) \hookrightarrow \alg{A}(\Lambda_2)$ of the corresponding observable algebras, by tensoring with identity matrices. If $\Lambda_1 \cap \Lambda_2 = \emptyset$, one has the \emph{locality} property: $[\alg{A}(\Lambda_1), \alg{A}(\Lambda_2)] = \{0\}$ (viewed as subalgebras of $\alg{A}(\Lambda)$ for some sufficiently large set $\Lambda$). This makes the assignment $\Lambda \mapsto \alg{A}(\Lambda)$ into a \emph{local net}. The \emph{quasi-local algebra} $\alg{A}$ is obtained by taking the union of all the local algebras, and completing with respect to the natural norm to obtain a $C^*$-algebra.\footnote{More abstractly, one can take the inductive limit of the net in the category of $C^*$-algebras.} Similarly, if $\Lambda$ is any (possibly infinite) subset of $\bonds$, the corresponding quasi-local algebra is defined by
\[
\alg{A}(\Lambda) = \overline{\bigcup_{\Lambda_f \subset \Lambda} \alg{A}(\Lambda)}^{\| \cdot \|}.
\]
The union is over all finite subsets of $\Lambda$ and the bar denotes closure in the norm topology.

So far the set-up has been very general. To recover the usual Hilbert space description of quantum mechanics we consider representations of $\alg{A}$. A $C^*$-algebra $\alg{A}$ generally has many different inequivalent representations, and many of those are not physically relevant (for example because the energy is not bounded from below). A natural choice is to consider ground state representations with respect to the dynamics of some model we are interested in. Let us write $\pi_0$ for such a choice of ground state, where we take $\pi_0$ to be irreducible. Such a representation can be obtained by applying the GNS construction to a pure ground state. The main interest in the present paper are the elementary excitations or charges with respect to this ground state. In the spirit of the Doplicher-Haag-Roberts (DHR) program in algebraic quantum field theory~\cite{MR1405610}, these correspond to inequivalent irreducible representations, that look like the ground state representation outside certain regions. In particular, we will restrict our attention to representations satisfying the following \emph{selection criterion}:
\begin{equation}
	\pi_0 \upharpoonright \alg{A}(\Lambda^c) \cong \pi \upharpoonright \alg{A}(\Lambda^c),
	\label{eq:cselect}
\end{equation}
for all cones $\Lambda$, where $\Lambda^c$ denotes the complement of $\Lambda$ in the set of bonds. By a cone we mean the following: consider a point in the plane, and consider two semi-infinite lines emanating from it, with an angle smaller than $\pi$ between them. If we embed the bonds of the system in the plane in the canonical way, the cone $\Lambda$ consists of all those bonds which, when their endpoints are removed, intersect the area bounded by these two lines (see also Fig.~\ref{fig:paths} in Section~\ref{sec:toric}).

In words, the criterion~\eqref{eq:cselect} selects those representations of $\alg{A}$ that are unitarily equivalent to the ground state representation $\pi_0$ \emph{when restricted to observables outside \emph{arbitrary} cones $\Lambda$}. Note that this is similar to the selection criterion used in the context of massive particles in algebraic quantum field theory~\cite{MR660538}. The precise shape of the cones plays only a minor role (see also the comment after Definition~\ref{def:globalindex}), but this choice is easy to visualise.

In the present context, this selection criterion is motivated by the appearance of ``stringlike'' excitations in many models relevant for topological quantum computing~\cite{haahcph,MR1951039,PhysRevB.71.045110,toricendo}. In these models excitations are created in pairs, connected by a string. Such states depend only on the endpoints of the strings, hence one already observes some topological features. To consider a \emph{single} excitation, one can move one end of the pair to infinity. Because of the topological feature mentioned before, the result should not depend on the direction in which we move the particle away. This is precisely captured by the selection criterion~\eqref{eq:cselect}. This procedure has been worked out explicitly for the toric code model in~\cite{toricendo}.

The study of the charge superselection sectors hence amounts to studying the set of representations fulfilling the selection criterion. From a technical perspective it is however easier to work with \emph{endomorphisms} of the observable algebra $\alg{A}$, instead of representations, because endomorphisms carry more structure.\footnote{One can compose two endomorphisms, for example.} This can be achieved with the help of \emph{Haag duality}, whose definition we recall below. All physically relevant properties of the different charges, such as their braiding and fusion properties, can be recovered by the study of these endomorphisms.

Another property that we will need is the distal split property. For the convenience of the reader we recall these definitions:
\begin{definition}
	\label{def:conerel}
	Let $\Lambda \subset \bonds$. Then the von Neumann algebra $\mc{R}(\Lambda)$ is defined by $\mc{R}(\Lambda) = \pi_0(\alg{A}(\Lambda))''$. We say that $\pi_0$ satisfies \emph{Haag duality for cones} if for each cone $\Lambda \subset \bonds$ one has $\mc{R}(\Lambda) = \mc{R}(\Lambda^c)'$. Here $\Lambda^c$ is the complement of $\Lambda$ in the set of bonds $\bonds$. Finally, $\pi_0$ has the \emph{distal split property}, if there is a relation $\ll$ on the set of cones with the following properties: $\Lambda_1 \ll \Lambda_2$ implies $\Lambda_1 \subset \Lambda_2$ and if $\Lambda_1 \ll \Lambda_2$, then there is a Type I factor $\mc{N}$ such that $\mc{R}(\Lambda_1) \subset \mc{N} \subset \mc{R}(\Lambda_2)$.
\end{definition}
To avoid trivial cases we always assume that the relation $\ll$ is non-empty. In particular, we will assume that for each cone $\Lambda_1$ there is a disjoint cone $\Lambda_2$ such that there is a cone $\Lambda$ with $\Lambda_1 \ll \Lambda$ and $\Lambda_2 \subset \Lambda^c$. As will become clear, this is a natural requirement in view of the superselection criterion above. This demand can be weakened by restricting the set of admissible cones, if necessary. As will become clear later, the relation $\ll$ describes a kind of ``separation'' property on the observable algebras associated to cones. Our analysis requires both Haag duality and the distal split property, and from now on we will assume that $\pi_0$ has these properties. In the case of the toric code, there is a unique (and pure) ground state~\cite{MR2345476}. The corresponding GNS representation $\pi_0$ satisfies both Haag duality for cones and the distal split property~\cite{haagdtoric}.

It is instructive to recall how one can obtain endomorphisms from representations. Let $\pi$ be a representation satisfying the selection criterion~\eqref{eq:cselect}, and let $\Lambda$ be a cone. Then there is a unitary operator $V: \mc{H}_0 \to \mc{H}_{\pi}$ such that $V \pi_0(A) = \pi(A) V$ for all $A \in \alg{A}(\Lambda^c)$. Define $\rho(A) := V^* \pi(A) V$, which is a representation of $\alg{A}$ on $\mc{H}_0$. Now suppose that $\widehat{\Lambda}$ is a cone containing $\Lambda$ and let $A \in \alg{A}(\widehat{\Lambda})$ and $B \in \alg{A}(\widehat{\Lambda}^c)$. Then
\[
	\pi_0(B) \rho(A) = \pi_0(B) V^* \pi(A) V = V^* \pi(BA) V = V^* \pi(AB) V = \rho(A) \pi_0(B).
\]
It follows that $\rho(A) \in \mc{R}(\widehat{\Lambda}^c)'$, hence in $\mc{R}(\widehat{\Lambda})$ by Haag duality. Note that $\rho(A) = \pi_0(A)$ for all $A \in \alg{A}(\Lambda^c)$, i.e. it acts trivially on observables localised outside the cone $\Lambda$. We therefore say that $\rho$ is \emph{localised} in $\Lambda$. To obtain an endomorphism, one can show that in fact one can extend $\rho$ to an endomorphism of an auxiliary algebra $\alg{A}^{\Lambda_a}$, containing the algebras $\mc{R}(\widehat{\Lambda})$, c.f.~\cite{MR660538,toricendo}. The details are of only minor importance here: it is enough to restrict to endomorphisms of the form $\rho : \mc{R}(\Lambda) \to \mc{R}(\Lambda)$ satisfying transportability. Here \emph{transportability} means that for any other cone $\widehat{\Lambda}$, there is an endomorphism $\sigma$ localised in $\widehat{\Lambda}$ such that $\rho(A) = V \sigma(A) V^*$ for all $A \in \alg{A}$ and some unitary $V$. Since the selection criterion~\eqref{eq:cselect} holds for \emph{any} cone $\Lambda$, it follows that the corresponding endomorphisms as obtained above are transportable.

The index that we consider in the next section will give a bound on the number of such localised endomorphisms modulo inner automorphisms (of $\mc{R}(\Lambda)$). It is enough to consider irreducible morphisms, in the sense that $\rho(\mc{R}(\Lambda))' \cap \mc{R}(\Lambda) = \mathbb{C}$. If $\rho$ is obtained from an irreducible representation satisfying the selection criterion, it follows with the help of the localisation of $\rho$ in $\Lambda$ that $\rho$ is irreducible in this sense, when restricted to $\mc{R}(\Lambda)$. The converse is also true (c.f.~\cite{MR1181069}).

\section{Cone index}\label{sec:index}
In~\cite{toricendo} we identified a number of representations satisfying the selection criterion for the toric code. A natural question to ask is if these are in fact \emph{all} such representations. Following ideas of Kawahigashi, Longo and M\"uger~\cite{MR1838752}, who analysed the case of conformal field theories on the circle, we show that this question is related to the Kosaki-Longo index of certain inclusions of von Neumann algebras.

At this point we first describe the type of algebras that we consider. The algebras are associated to certain regions that are unions of cones.
\begin{definition}
We will write $\mathcal{C}^2$ for the set of subsets $\Lambda_1 \cup \Lambda_2 \subset \bonds$, where $\Lambda_1$ and $\Lambda_2$ are cones such that there is some cone $\Lambda$ with $\Lambda_1 \ll \Lambda$ and $\Lambda_2 \subset \Lambda^c$.
\end{definition}
Note that by the assumption made after Definition~\ref{def:conerel} the set $\mathcal{C}^2$ is non-empty. In general we will have $\bigcup_{\Gamma \in \mc{C}^2} \Gamma = \bonds$, but this assumption is not strictly necessary for our analysis. This condition ensures that we can apply the distal split property. It follows that for $\Gamma \in \mathcal{C}^2$, we have $\mc{R}(\Gamma) = \mc{R}(\Lambda_1) \vee \mc{R}(\Lambda_2) \simeq \mc{R}(\Lambda_1) \otimes \mc{R}(\Lambda_2)$ by the split property. The set $\mc{C}^2$ plays the same role as the set $\mc{I}_2$ of pairs of intervals (on the circle) with disjoint closure in~\cite{MR1838752}. Note however that if $\Gamma \in \mathcal{C}^2$, it is not true that the complement $\Gamma^c$ is in $\mathcal{C}^2$.

To each $\Gamma \in \mc{C}^2$ we will assign an irreducible subfactor. To this end, define $\widehat{\mc{R}}(\Gamma) := \mc{R}(\Gamma^c)'$. Note that by locality we have $\mc{R}(\Gamma) \subset \widehat{\mc{R}}(\Gamma)$. If $\Gamma$ would be a cone it follows by Haag duality that $\mc{R}(\Gamma) = \widehat{\mc{R}}(\Gamma)$, but for the union of two cones this inclusion is in general non-trivial. The next Lemma shows that this inclusion is indeed an irreducible subfactor.
\begin{lemma}\label{lem:factors}
	Let $\Gamma \in \mathcal{C}^2$. Then $\mc{R}(\Gamma) \subset \widehat{\mc{R}}(\Gamma)$ is an irreducible inclusion of factors.
\end{lemma}
\begin{proof}
Let us first show that both von Neumann algebras are indeed factors. Since $\pi_0$ is an irreducible representation, it follows that $\mc{R}(\Gamma) \vee \mc{R}(\Gamma^c) = \alg{B}(\mc{H})$. Note that we have that
\[
	\left( \mc{R}(\Gamma)' \cap \widehat{\mc{R}}(\Gamma) \right)' = \mc{R}(\Gamma) \vee \mc{R}(\Gamma^c) = \alg{B}(\mc{H}).
\]
It follows that $\mc{R}(\Gamma)' \cap \widehat{\mc{R}}(\Gamma) = \mathbb{C} I$, that is, the inclusion is irreducible. It then automatically follows that $\mc{R}(\Gamma)$ and $\widehat{\mc{R}}(\Gamma)$ are factors.
\end{proof}

As mentioned before, we assume that $\mc{R}(\Gamma)$ and $\widehat{\mc{R}}(\Gamma)$ are infinite factors. Note that this is essentially an additional assumption on the representation $\pi_0$. Nevertheless, it is true under quite general (and natural) assumptions. For example, if $\pi_0$ is the GNS representation of a pure translational invariant state $\omega_0$, it follows that $\mc{R}(\Gamma)$ is infinite, since it is (with the split property) naturally isomorphic to $\mc{R}(\Lambda_1) \otimes \mc{R}(\Lambda_2)$, both of the tensor factors being infinite factors. The latter can be seen from the proof of~\cite[Thm. 5.1]{toricendo} (c.f.~\cite[Prop. 5.3]{MR2281418}). In any case, we make this assumption only for simplicity. 

The previous Lemma provides us with irreducible subfactors and leads naturally to our main numerical quantity of interest: the Kosaki-Longo index associated to inclusions $\mc{R}(\Gamma) \subset \widehat{\mc{R}}(\Gamma)$. In principle this index can depend on the choice of $\Gamma$. However, the next lemma shows that it is monotonically increasing depending on the size of $\Gamma$.
\begin{proposition}
Let $\Gamma$ and $\widehat{\Gamma}$ be elements of $\mc{C}^2$, such that $\Gamma \subset \widehat{\Gamma}$. Then $[\widehat{\mc{R}}(\Gamma) : \mc{R}(\Gamma)] \leq [\widehat{\mc{R}}(\widehat{\Gamma})  : \mc{R}(\widehat{\Gamma})]$.
\end{proposition}
\begin{proof}
	We adapt the proof of~\cite[Prop. 5]{MR1838752} to the case at hand. Suppose that $\Gamma = \Lambda_1 \cup \Lambda_2 \in \mc{C}^2$, and choose a cone $\widehat{\Lambda}_2 \supset \Lambda_2$ such that $\widehat{\Gamma} = \Lambda_1 \cup \widehat{\Lambda}_2 \in \mc{C}^2$. Without loss of generality we may assume that $\lambda^{-1} = [\widehat{\mc{R}}(\widehat{\Gamma}) : \mc{R}(\widehat{\Gamma})] < \infty$. Let $\mc{E}_{\widehat{\Gamma}}$ be the corresponding conditional expectation. We claim that $\mc{E}_{\widehat{\Gamma}}$ restricts to a conditional expectation of $\widehat{\mc{R}}(\Gamma)$ onto $\mc{R}(\Gamma)$, hence 
\begin{equation}
	\label{eq:condineq}
	[\widehat{\mc{R}}(\Gamma) : \mc{R}(\Gamma)] \leq [\widehat{\mc{R}}(\widehat{\Gamma})  : \mc{R}(\widehat{\Gamma})]
\end{equation}
by the Pimsner-Popa inequality. It is enough to show that $\mc{E}_{\widehat{\Gamma}}$ maps $\widehat{\mc{R}}(\Gamma)$ into $\mc{R}(\Gamma)$. Note that $\widehat{\mc{R}}(\Gamma) \subset \mc{R}(\Lambda_2^c \cap \widehat{\Lambda}_2)'$. Since $\mc{E}_{\widehat{\Gamma}}$ acts trivially on $\mc{R}(\Lambda_2^c \cap \widehat{\Lambda}_2)$, it follows that
\[
\mc{E}_{\widehat{\Gamma}}(\widehat{\mc{R}}(\Gamma)) \subset \mc{R}(\Lambda_2^c \cap \widehat{\Lambda}_2)' \cap \mc{R}(\widehat{\Gamma}).
\]
We claim that the right-hand set is equal to $\mc{R}(\Gamma)$, from which equation~\eqref{eq:condineq} follows. To show this, note that by the split property we have 
\[
	\mc{R}(\widehat{\Gamma}) \simeq \mc{R}(\Lambda_1) \otimes \mc{R}(\widehat{\Lambda}_2).
\]
By locality it follows that $\mc{R}(\Lambda_1) \subset \mc{R}(\Lambda_2^c \cap \widehat{\Lambda}_2)'$. Moreover,
\[
\begin{split}
\mc{R}(\widehat{\Lambda}_2) &\cap \mc{R}(\Lambda_2^c \cap \widehat{\Lambda}_2)' = (\mc{R}(\widehat{\Lambda}_2)' \vee \mc{R}(\Lambda_2^c \cap \widehat{\Lambda}_2))'\\ &= (\mc{R}(\widehat{\Lambda}_2^c) \vee \mc{R}(\Lambda_2^c \cap \widehat{\Lambda}_2))' = \mc{R}(\Lambda_2^c)' = \mc{R}(\Lambda_2),
\end{split}
\]
where we used Haag duality twice. The result then follows. 
\end{proof}

\begin{remark}
	This should be contrasted with the case of rational conformal field theory on the circle~\cite{MR1838752}. There one considers two open intervals on the circle with disjoint closures, and associated algebras (in our notation) $\widehat{\mc{R}}(\Gamma)$ and $\mc{R}(\Gamma)$, with $\Gamma = I_1 \cup I_2$. Using the same proof as in the proposition above it follows that the index increases by increasing the intervals. However, in their case (the interior of) $\Gamma^c$ again consists of two open intervals. Using $[\widehat{\mc{R}}(\Gamma): \mc{R}(\Gamma)] = [\mc{R}(\Gamma)':\widehat{\mc{R}}(\Gamma)']$. The index at the right hand side is in fact the index associated to the two intervals in $\Gamma^c$ and one can conclude that the index is constant. In the present situation, however, $\Gamma^c$ is not the union of two cones that are sufficiently separated, hence the split property cannot be used any more.
\end{remark}

\begin{definition}
	\label{def:globalindex}
Let $\pi_0$ be a representation as above. Then we define the \emph{cone index} by $\mu_{\pi_0} = \inf_{\Gamma \in \mc{C}^2} [\widehat{\mc{R}}(\Gamma) : \mc{R}(\Gamma)]$.
\end{definition}
The restriction to regions of the from $\Gamma = \Lambda_1 \cup \Lambda_2$ is likely not essential. By looking at the proofs of Haag duality and the split property in the toric code~\cite{haagdtoric}, as well as at the calculation of $\mu_{\pi_0}$ in the next section, it is apparent that the geometry of the cones is not important. What is needed for the regions $\Lambda_1$ and $\Lambda_2$ are (at least) the following properties. First of all, it should be possible to choose (dual) paths to infinity in the cones, per the construction of the different superselection sectors~\cite{toricendo}. Secondly, they should be ``wide'' enough in the sense that any finite set of sites can be translated into the region. This ensures that the factors $\mc{R}(\Gamma)$ are infinite.\footnote{The index can be defined for finite factors as well (and indeed, it was first defined for subfactors of Type II$_1$), but there are technical differences between these cases.} And finally the regions $\Lambda_1$ and $\Lambda_2$ should not contain any ``holes'', so that there is only a boundary at the ``outside'' of the region. Hence the definition can be generalised,\footnote{The superselection criterion~\eqref{eq:cselect} should then be suitably modified as well.} but (at least for the toric code) we would not gain anything from it.

Under suitable conditions one can define a conjugate for a charge sector, say represented by a localised endomorphism $\rho$, which we assume to be irreducible for simplicity. Such a conjugate $\overline{\rho}$ defines a numerical quantity, the statistical dimension $d(\rho)$, for the sector. If a conjugate can be defined, this number is finite, and in fact one has the remarkable relation $d(\rho)^2 = [\mc{R}(\Lambda) : \rho(\mc{R}(\Lambda))]$, if $\Lambda$ is a cone containing the localisation region of $\rho$. This is well-known in algebraic quantum field theory~\cite{MR1027496,MR1059320}, but it is true in the present context as well. First note that $\rho$ is injective, hence $\rho(\mc{R}(\Lambda))$ is a factor, which by localisation and Haag duality is contained in $\mc{R}(\Lambda)$. For our purposes it is sufficient to restrict to the case where $\rho \upharpoonright \mc{R}(\Lambda)$ is irreducible: $\rho(\mc{R}(\Lambda))' \cap \mc{R}(\Lambda) = \mathbb{C} I$.  Hence in that case $\rho(\mc{R}(\Lambda)) \subset \mc{R}(\Lambda)$ is an irreducible subfactor. The claim then follows by constructing a suitable conditional expectation in the following way.

By the definition of a conjugate~\cite{MR1444286}, there are intertwiners $R, \overline{R}$ such that $R A = \overline{\rho} \circ \rho(A) R$ and similarly for $\overline{R}$. If such intertwiners exist, they can be chosen to satisfy $R^*R = \overline{R}^* \overline{R} := d(\rho) I$ and $\overline{R}^*\rho(R) = R^* \overline{\rho}(\overline{R}) = I$. The number $d(\rho)$ is called the \emph{statistical dimension} of $\rho$, and does not depend on the choice of representative for the sector $\rho$.

Now define a map $\mc{E} : \mc{R}(\Lambda) \to \rho(\mc{R}(\Lambda))$ by
\[
\mc{E}(A) := d(\rho)^{-1} \rho(R^* \overline{\rho}(A) R).
\]
This is a normal completely positive map. Moreover, note that $\mc{E}(\rho(A)) = \rho(R^* \overline{\rho}(\rho(A)) R) = \rho(A)$, since $\overline{\rho}$ is a conjugate. Because $\overline{\rho}$ is an endomorphism, the bimodule property is easily verified, and $\mc{E}$ is seen to be a conditional expectation. One can show that $\mc{E}(A^*A) \geq d(\rho)^{-2} A^*A$ (see e.g.~\cite[Sect. 3]{MR1444286}). The bound is satisfied by the projection $P_\rho = d(\rho)^{-1} \overline{R}\, \overline{R}^*$. It follows that $[\mc{R}(\Lambda) : \rho(\mc{R}(\Lambda))] = d(\rho)^2$.

Conversely, if this index is finite, one can \emph{define} a conjugate $\overline{\rho}$~\cite{MR1059320}, where $\overline{\rho} : \mc{R}(\Lambda) \to \mc{R}(\Lambda)$. However, this does not immediately give a localised and transportable endomorphism of $\alg{A}$, and hence a conjugate in the sense above. Essentially, what has to be shown is that the conjugate can be defined consistently for a larger localisation region $\widehat{\Lambda} \supset \Lambda$, and that by choosing an equivalent $\widehat{\rho} \cong \rho$ localised in some other cone, the resulting conjugates are unitarily equivalent as well. We will not investigate the issue in more detail here, since it is not necessary for our purposes, but see e.g.~\cite{MR1181069,MR1332979} for a similar problem in the context of algebraic quantum field theory.

We are now in a position to state the main result: the cone index $\mu_{\pi_0}$ allows us to give bounds on the number of superselection sectors.
\begin{theorem}
	\label{thm:upperbound}
The number of equivalence classes of representations satisfying the selection criterion is bounded from above by $\lfloor \mu_{\pi_0} \rfloor$. If each sector has a conjugate, one has in fact that $\sum_i d(\rho_i)^2 \leq \mu_{\pi_0}$, where the sum is over the different equivalence classes of cone-localised transportable endomorphisms, and $d(\rho_i)$ is the statistical dimension of a representative $\rho_i$.
\end{theorem}
\begin{proof}
	The first claim can be proven in essentially the same way as Lemma 13 of~\cite{MR1838752}. As for the inequality, without loss of generality we can assume $\mu_{\pi_0}$ to be finite. One can then check that the proof of Corollary 10 in~\cite{MR1838752} works in the present setting as well.
\end{proof}

The bound relies on the assumption that for any pair $\Lambda_1 \cup \Lambda_2 \in \mathcal{C}$ and any sector $[\rho]$, we can find a representative localised in $\Lambda_1$ and in $\Lambda_2$. This follows from the superselection criterion~\eqref{eq:cselect} and Haag duality. It may be conceivable that there are additional sectors for which this does not hold, for example because they can only be localised in cones with a certain minimum opening angle. There are two ways to handle such cases: one can either shrink the set $\mc{C}$ of admissible pairs of cones, or one can consider a supremum instead of an infimum in Definition~\ref{def:globalindex}. Note that in the context of rational CFT~\cite{MR1838752} (with the additional assumption that there is a ``modular PCT''), the index is equal to the sum of the squares of the dimension of each sector.\footnote{In the language of fusion categories, this means that the index is equal to the dimension of the fusion category of superselection sectors.} However, this result depends on the charge conjugation being related to the so-called \emph{modular PCT symmetry}. At the moment we do not have a comparable result in the setting of spin systems.

\begin{remark}
	We would like to stress that the analysis can in principle be done for \emph{any} model satisfying the assumptions, the most important ones being Haag duality and the distal split property for cones. Note for example that (local) Hamiltonians do not appear directly. One should note however that the representation $\pi_0$, being a ground state representation, \emph{does} strongly depend on the dynamics. Hence the Hamiltonian comes in precisely when selecting the representation $\pi_0$. In principle one can do the analysis with any representation $\pi_0$ satisfying the assumptions (that is, not necessarily a ground state representation), but the physical interpretation is less clear in that case. An interesting open question is to classify all the (local) Hamiltonians whose corresponding ground state representations satisfy the conditions of Haag duality and the distal split property, and in addition give rise to non-trivial representations satisfying the selection criterion~\eqref{eq:cselect}.
\end{remark}

\section{The toric code}\label{sec:toric}
In this section we apply the methods developed above to a specific model: Kitaev's toric code~\cite{MR1951039}. Again, we consider the model on an infinite 2D lattice~\cite{MR2345476,toricendo}. The assumptions of Haag duality and the distal split property for cones can be shown to hold in this model~\cite{haagdtoric}. We first recall the most important features and definitions~\cite{MR1951039,toricendo}.

For our purposes it is enough to define the model on a square lattice, as in Sect.~\ref{sec:charges}. At each bond of the lattice there is a qubit, hence the observables at each bond are the $2\times2$ matrices $M_2(\mathbb{C})$. A star $s$ consists of bonds sharing the same vertex of the lattice as their endpoints. Similarly, a plaquette consists of those four bonds that  enclose a square. To each star $s$ and plaquette $p$ we associate the operators
\[
A_s := \bigotimes_{j \in s} \sigma_j^x, \quad  B_p := \bigotimes_{j \in p} \sigma_j^z.
\]
With this notation we can introduce the local Hamiltonians $H_{\Lambda}$, describing the dynamics. If $\Lambda \in \mc{P}_f(\bonds)$, then
\[
	H_\Lambda = - \sum_{s \subset \Lambda} A_s - \sum_{p \subset \Lambda} B_p. 
\]	
The model has a unique pure ground state $\omega_0$ for these dynamics~\cite{MR2345476}. We denote $(\pi_0, \mc{H}, \Omega)$ for the corresponding GNS representation. Since $\alg{A}$ is simple, $\pi_0$ is injective and we will often identify $\pi_0(A)$ with $A$. The ground state is invariant under the action of the star and plaquette operators, $A_s \Omega = B_p \Omega = \Omega$.

Let $\xi$ be a set of bonds forming a path on the lattice. Let $\widehat{\xi}$ be a path on the dual lattice. We will identify $\widehat{\xi}$ with all the bonds that the path on the dual lattice crossed. For finite paths (or dual paths), we introduce the operators
\[
F_{\xi} = \bigotimes_{j \in \xi} \sigma^z_j, \quad F_{\widehat{\xi}} = \bigotimes_{j \in \widehat{\xi}} \sigma^z_j.
\]
It is easy to check that these operators commute with all star and plaquette operators, except with those stars which are based at the endpoints of the path $\xi$, or those plaquette operators associated to the endpoints of $\widehat{\xi}$. It is then easy to check that $F_\xi \Omega$ is again an eigenvector of the local Hamiltonians, and $F_\xi \Omega$ is interpreted as a state with two excitations at the endpoints of $\xi$.

Two important properties of these path operators that we will frequently use are that $F_{\xi} \Omega = F_{\xi'} \Omega$ if the (dual) paths $\xi$ and $\xi'$ have the same endpoints, and that $F_\xi \Omega = \Omega$ if $\xi$ is a loop. The distinction between paths and dual paths is often not important, hence we will usually not distinguish them in the sequel. If $\Lambda \subset \bonds$ is any subset of bonds, we write $\mathfrak{p}(\Lambda)$ for the set of finite paths and dual paths inside $\Lambda$.

To illustrate the general theory of the previous section, the aim here is to calculate the index $[ \widehat{\mc{R}}(\Gamma) : \mc{R}(\Gamma) ]$ for this model. This requires a good understanding of the algebra $\widehat{\mc{R}}(\Gamma)$. Using similar techniques as developed in~\cite{haagdtoric}, we can prove that this algebra can easily be described in terms of $\mc{R}(\Gamma)$ and the charge transporters constructed in~\cite{toricendo}. We briefly recall their construction. Let $\gamma_1$ ($\gamma_2$) be paths to infinity inside $\Lambda_1$ ($\Lambda_2$). For each $n \in \mathbb{N}$, write $\gamma_i(n)$ for the finite path consisting of the first $n$ parts of $\gamma_i$, and choose a path $\widehat{\gamma}_n$ connecting the $n$-th site of $\gamma_1$ with the $n$-th site of $\gamma_2$, such that the distance of the paths $\widehat{\gamma}_n$ to the endpoints of $\gamma_1$ and $\gamma_2$ tends to infinity as $n$ goes to infinity. We write $\gamma_1(n)\widehat{\gamma}_n\gamma_2(n)$ for the concatenated path. Then the sequence of operators $F_{\gamma_1(n)\widehat{\gamma}_n \gamma_2(n)}$ converges weakly to a unitary $V_X$, the ``charge transporter'' that transports a charge of type X from $\Lambda_1$ to $\Lambda_2$. In a similar way we obtain a unitary $V_Z$ by considering paths on the dual lattice. These unitaries will be fixed for the remainder of the paper. With this notation, the crucial lemma can be stated as follows.

\begin{lemma}
	\label{lem:vnagen}
Let $\Gamma \in \mathcal{C}^2$ with $\Gamma = \Lambda_1 \cup \Lambda_2$. Then $\widehat{\mc{R}}(\Gamma)$ is generated by $\mc{R}(\Gamma)$ and unitaries $V_X$ and $V_Z$. Here $V_X$ transports a charge of type $X$ from $\Lambda_1$ to $\Lambda_2$, and similarly for $V_Z$.
\end{lemma}

The proof of this lemma is done in a number of steps. The essential idea is as follows. Let us write $\mc{A}$ for the von Neumann algebra $\mc{R}(\Gamma) \vee \{V_X, V_Z\}$. Then we will show that $\mc{A}$ is the commutant of $\mc{B} := \alg{A}(\Gamma^c)$ (hence $\mc{B}' = \widehat{\mc{R}}(\Gamma)$). The problem becomes tractable by first restricting the algebras of interest to a subspace of the ground state Hilbert space $\mc{H}_{\Lambda}$.

\begin{definition}
	Let $\Gamma = \Lambda_1 \cup \Lambda_2$ be an element of $\mathcal{C}^2$. We define the subspace $\mc{H}_{\overline{\Gamma}}$ as the closed linear span of $\mc{A}\Omega$, where $\mc{A}$ is as defined above.
\end{definition}

A subset of $\mc{H}_{\overline{\Gamma}}$ whose linear span is dense in $\mc{H}_{\overline{\Gamma}}$ can be obtained by taking the ground state $\Omega$ and create pairs of excitations, where each end of the pair is contained in (or at the boundary of) $\Gamma$. It is clear that if both of the excitations of such a pair lie in the same cone, there is some path operator $F_\xi \in \alg{A}(\Gamma)$ creating this excitation. With the help of $V_X$ ($V_Z$) it is possible to create excitations of type X (Z) that lie in the two distinct cones. To formalise this, we first introduce new notation.

\begin{figure}
	\includegraphics[width=\textwidth]{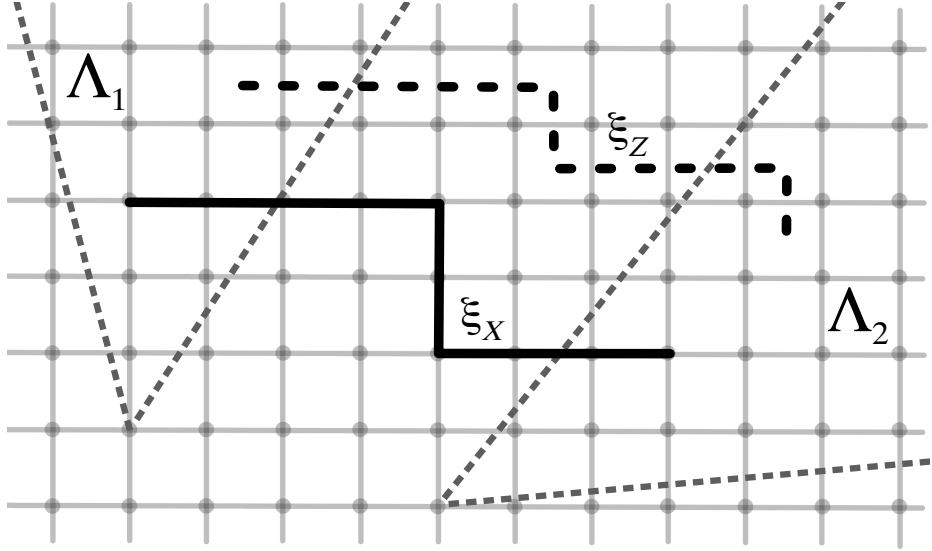}
	\caption{Cones $\Lambda_1$ and $\Lambda_2$, together with the choice of a path $\xi_{Z}$ and a dual path $\xi_{X}$.}
	\label{fig:paths}
\end{figure}

\begin{definition}
	Write $\Gamma = \Lambda_1 \cup \Lambda_2$ and fix a path (dual path) $\xi_Z$ ($\xi_X$) from a vertex (plaquette) in $\Lambda_1$ to a vertex (plaquette) in $\Lambda_2$, see Fig.~\ref{fig:paths}. Then $\mc{F}(\Gamma)$ is defined as the set of all $F_{\xi}$ with $\xi \in \mathfrak{p}(\Gamma)$. Similarly,
\[
	\mc{F}(\overline{\Gamma}) := \{ F_\xi : \xi \in \mathfrak{p}(\Gamma) \cup \{\xi_X, \xi_Z \} \}.
\] 
In addition we set
\[
\alg{F}(\Gamma) := \{ F_{1} \dots F_{n} : F_i \in \mc{F}(\Gamma) \},
\]
and likewise for $\alg{F}(\overline{\Gamma})$.
\end{definition}
The choice of $\xi_X$ and $\xi_Z$ will be fixed for the remainder of this section.

Intuitively, the operators in $\mc{F}(\Gamma)$ create pairs of excitations that are both in either $\Lambda_1$ or $\Lambda_2$, while $\mc{F}(\overline{\Gamma})$ contains operators that create a pair of excitations in either cone. Such operators generate the whole Hilbert space $\mc{H}_{\overline{\Gamma}}$, as is demonstrated in the next lemma. In particular it follows that the definition of this space does not depend on the specific choice of paths in the definition of the unitaries $V_X$ and $V_Z$.

\begin{lemma}
\label{lem:clspan}
The space $\mc{H}_{\overline{\Gamma}}$ is equal to the closure of the span of $\alg{F}(\overline{\Gamma}) \Omega$.
\end{lemma}
\begin{proof}
In the construction of $V_Z$, one considers a path to infinity in $\Lambda_1$, and one in $\Lambda_2$. Choose a path $\xi_1$ in $\Lambda_1$, from the endpoint of the (semi-)infinite path in $\Lambda_1$ to the endpoint of the path $\xi_X$ on the boundary of $\Lambda_1$. Similarly, we can choose a path $\xi_2$ in $\Lambda_2$. It follows that the path $\xi_1 \cup \xi_Z \cup \xi_2$ connects the endpoints of the two semi-infinite paths. By the choice of $V_Z$ we have $F_{\xi_1} F_{\xi_Z} F_{\xi_2} V_Z \Omega = \Omega$ (cf. Lemma 4.1 of ~\cite{toricendo}) and hence $V_Z \Omega = F_{\xi_1} F_{\xi_Z} F_{\xi_2} \Omega$. Similarly, $V_X \Omega = F_1 F_{\xi_X} F_2 \Omega$ for suitable path operators $F_1, F_2 \in \mc{F}(\Gamma)$.

Next we show that $V_X$ and $V_Z$ either commute or anti-commute with a given path operator $F \in \mc{F}(\Gamma)$. Consider a finite path (or dual path) $\xi$ in $\Gamma$ and let $F_\xi$ be the corresponding path operator. Then clearly $F_\xi F_{\gamma_1(n)\gamma_n\gamma_2(n)} = \pm F_{\gamma_1(n)\gamma_n\gamma_2(n)} F_\xi$, and since $F_\xi$ is local, this sign is the same for all large enough $n$. Hence, by separate weak continuity of multiplication, it follows that
\begin{equation}
	\label{eq:chcommute}
V_Z F_{\xi} = \wlim_{n \to \infty} F_{\gamma_1(n) \gamma_n \gamma_2(n)} F_{\xi} = \pm \wlim F_{\xi} F_{\gamma_1(n) \gamma_n \gamma_2(n)} = \pm F_{\xi} V_Z.
\end{equation}
It follows that $F V_Z = \pm V_Z F$ for every $F \in \alg{F}(\Gamma)$. Note moreover that $V_X^2 = V_Z^2 = I$. The claim then follows, by letting operators of the form $F V_i$ with $F \in \alg{F}(\Gamma)$ act on $\Omega$, since these generate a dense subalgebra of $\mc{A} = \mc{R}(\Gamma) \vee \{V_X, V_Z\}$.
\end{proof}

Recall that we are essentially studying a commutation problem. With help of the next lemma the problem can be reduced to studying operators acting on the \emph{subspace} $\mc{H}_{\overline{\Gamma}}$.
\begin{lemma}\label{lem:restriction}
The algebras $\mc{A}$ and $\mc{B}'$ can be restricted to the subspace $\mc{H}_{\overline{\Gamma}}$, and each element of these algebras is completely determined by this restriction.
\end{lemma}
\begin{proof}
	It is immediate from the definition of $\mc{H}_{\overline{\Gamma}}$ that $\mc{A}$ can be restricted to this subspace. To show a similar claim for $\mc{B}'$, let $B \in \mc{B}'$. We will show that for $\zeta \in \mc{H}_{\overline{\Gamma}}$, we have $B \zeta \perp \eta$ for all $\eta \in \mc{H}_{\overline{\Gamma}}^\perp$. Let $\zeta = F_1 \cdots F_k V_X^{\sigma(x)} V_Z^{\sigma(z)} \Omega$, where $F_i \in \mc{F}(\Gamma)$ and $\sigma(x), \sigma(z) \in \{0,1\}$. Note that such vectors span a dense subset of $\mc{H}_{\overline{\Gamma}}$. Let $\eta = \widehat{F}_1 \cdots \widehat{F}_n F \Omega$, with $\widehat{F}_i \in \mc{F}(\Gamma^c)$ and $F \in \alg{F}(\overline{\Gamma})$.

	First note that $A_s \zeta = B_p \zeta = \zeta$ for all $\zeta \in \mc{H}_{\overline{\Gamma}}$ and $A_s, B_p$ elements of $\alg{A}(\Gamma^c)$, since such operators commute with any $F \in \alg{F}(\overline{\Gamma})$. Suppose moreover that there is such an $A_s$ such that $\{A_s, \widehat{F}_1 \cdots \widehat{F}_k\} = 0$. Then
\[
	(\eta, B \zeta) = (\eta, B A_s \zeta) = (A_s \eta, B \zeta) = - (\eta, B \zeta),
\]
hence $\eta \perp B \zeta$. The same is true if there is such a $B_p$.

If no such $A_s$ or $B_p$ exist, $\eta$ can only contain excitations at the boundary (or inside) the cones $\Gamma$, since otherwise there would be an anti-commuting star or plaquette operator in $\Gamma^c$. But then in fact $\eta \in \mc{H}_{\overline{\Gamma}}$. Pairs of excitations in a single cone can obviously be obtained in this way, and a pair where one of the excitations is in $\Lambda_1$ and the other in $\Lambda_2$ can be obtained by acting with $V_X$ (or $V_Z$) on $\Omega$. But this means that $\eta \in \mc{H}_{\overline{\Gamma}}$.\footnote{The reader who is uneasy with this brief justification is referred to the proof of Lemma~\ref{lem:sadjoint} for similar considerations, or to the proof of~\cite[Lemma 3.6]{haagdtoric}, which is alike.}

Finally, note that linear span of vectors of the form $F \eta$, with $\eta \in \mc{H}_{\overline{\Gamma}}$ and $F \in \alg{F}(\Gamma^c)$ is dense in $\mc{H}$. Now suppose that $A_1$ and $A_2$ agree on $\mc{H}_{\overline{\Gamma}}$, with $A_1, A_2$ in $\mc{A}$ or $\mc{B}'$. Then we have
\[
	A_1 F \eta = F A_1 \eta = F A_2 \eta = A_2 F \eta. 
\]
Since $A_1$ and $A_2$ are bounded linear operators, it follows that $A_1 = A_2$. This proves the claim that the operators are determined by their restrictions.
\end{proof}
The proof implies that $P_{\overline{\Gamma}}$ is an element of $\mc{B}'' = \mc{R}(\Gamma^c)$. We will write $\widehat{\mc{B}} := P_{\overline{\Gamma}} \mc{B}'' P_{\overline{\Gamma}}$ for $\mc{B}''$ restricted to $\mc{H}_{\overline{\Gamma}}$, and similarly for $\widehat{\mc{A}}$.

The crucial ingredient of the proof of Lemma~\ref{lem:vnagen} is the following result, which allows us to invoke a theorem of Rieffel and Van Daele~\cite{MR0383096} on the commutant of von Neumann algebras.
\begin{lemma}
	\label{lem:sadjoint}
	The linear space $(\widehat{\mc{A}}_s + i \widehat{\mc{B}}_s) \Omega$ is dense in $\mc{H}_{\overline{\Gamma}}$, where $\widehat{\mc{A}}_s$ is the self-adjoint part of $\widehat{\mc{A}}$.
\end{lemma}
\begin{proof}
	By Lemma~\ref{lem:clspan} it is enough to show that we can write vectors of the form $F \Omega$, with $F \in \mc{F}(\overline{\Gamma})$, in the desired form. Suppose that we are given such a vector $\eta = F \Omega$. Then there are $F_i \in \mc{F}(\Gamma)$ such that $F \Omega = A\Omega$ where $A := F_1 \dots F_k V_X^{\sigma(x)} V_Z^{\sigma(z)}$ with $\sigma(x),\sigma(z) \in \{0,1\}$. Clearly, $A \in \widehat{\mc{A}}$ and by the commutation properties of the $F_i$ and $V_k$ we have $A^* = \pm A$. If $A = A^*$ we can get real multiples of the vector. In the other case, $i A$ is self-adjoint and we get imaginary multiples of $\eta$. To get imaginary multiples (respectively real multiples, if $A^* = - A$), consider the case where there is a star $s$ in $\Gamma$ such that $A_s A = - A A_s$. Then either $i A_s A$ or $A_s A$ is self-adjoint, and $A_s A \Omega = - A \Omega$. Hence we can get real and imaginary multiples of $F \Omega$. The case of an anti-commuting plaquette operator is handled, mutatis mutandis, in the same way.

	Remains the case where such anti-commuting stars or plaquettes do not exist. Consider for convenience the case $A=A^*$: the goal is to find $B \in \widehat{\mc{B}}_s$ such that $i B \Omega = i A \Omega$. Since there are no anti-commuting plaquette and star operators inside $\Gamma$, $A$ can only create excitations at the boundary of the two cones. Since the path operators always create pairs of excitations, their total number must be even. We can then join pairs of these excitations by paths $\xi_i$ in $\Gamma^c$, and consider the corresponding path operators $F_i \in \mc{B}_s$. The question is if we can choose the paths $\xi_i$ in such a way that $\prod_i F_i$ is self-adjoint. This can only fail if there is an \emph{odd} number of crossings of paths and dual paths, since this would lead to an odd number of minus signs when commuting the path operators. But this cannot be the case, since it must then follow that $A^* = -A$, a contradiction. 
	
To see why this is the case, write $B = \prod_i F_i$. Since $A$ and $B$ commute (by locality), we have $(AB)^* = A^* B^* = A B^*$. But $B^* = \pm B$, so the conclusion follows if we can show that $AB$ is in fact self-adjoint. Since by construction all excitations are connected to another excitation by paths $\xi_i$, there are no free excitations left, and $AB$ is a product of path operators corresponding to loops, where we regard the case where $A$ contains a factor $V_X$ or $V_Z$ as an ``infinite'' loop. Note that for a finite loop the corresponding path operator is a product of star or plaquette operators, hence in that case $AB$ is self-adjoint. In the case there are ``infinite'' loops, it follows that $AB$ is the weak limit of a sequence of operators corresponding to closed loops, and the result follows by continuity of the $*$-operation with respect to the weak operator topology.
	
The remaining case where $A = - A^*$ (and hence $i A$ is self-adjoint) is handled in a similar way.
\end{proof}

This gives us all the necessary ingredients for the proof of the main Lemma of this section.
\begin{proof}[Proof of Lemma~\ref{lem:vnagen}]
	The proof is essentially the same as that of~\cite[Thm. 3.1]{haagdtoric}. We recall the gist of the argument. By Lemma~\ref{lem:sadjoint} and~\cite[Thm. 2]{MR0383096}, $\widehat{\mc{A}}'' = \widehat{\mc{B}}'$. Let $B \in \widehat{\mc{R}}(\Gamma)$. Then $B$ restricts to $\mc{H}_{\overline{\Gamma}}$, and its restriction $B_{\overline{\Gamma}}$ is in $\widehat{\mc{B}}'$ (c.f.~\cite[Prop. II.3.10]{MR1873025}). But this means that $B_{\overline{\Gamma}} = A_{\overline{\Gamma}}$ for some $A_{\overline{\Gamma}} \in \widehat{\mc{A}}''$. This $A_{\overline{\Gamma}}$ is in turn the restriction of some $A \in \mc{A}$. A straightforward calculation shows that we must in fact have $A = B'$, proving that $\widehat{\mc{R}}(\Gamma) \subset \mc{R}(\Gamma) \vee \{V_X, V_Z\}$. The reverse inclusion follows from locality and localisation (that is, the fact that the $V_k$ transport charges localised in $\Lambda_1 \subset \Gamma$ to $\Lambda_2 \subset \Gamma$).
\end{proof}

We now have a description of $\widehat{\mc{R}}(\Gamma)$ in terms of $\mc{R}(\Gamma)$ and the charge transporters. It turns out that in fact there is an even more explicit description of $\widehat{\mc{R}}(\Gamma)$, which clarifies the inclusion $\mc{R}(\Gamma) \subset \widehat{\mc{R}}(\Gamma)$. The key point is that we can consider the unitaries $V_X, V_Y$ and $V_Z$ as a representation of some finite group, and that these unitaries induce an action on $\mc{R}(\Gamma)$. It turns out that $\widehat{\mc{R}}(\Gamma)$ is isomorphic to the crossed product of $\mc{R}(\Gamma)$ by the action of this group. This would follow for example if this action were outer, by~\cite[Exercise X.1.1]{MR1943006}, but this is not immediately clear. We will therefore give such an isomorphism explicitly.

Let us first clarify how the charge transporters give rise to an action. First of all, without loss of generality we may assume that $V_X$ and $V_Z$ commute.\footnote{In the case they \emph{anti}-commute, we obtain merely a \emph{projective} representation. The induced action, and hence the crossed product, does not change, but the proof of Lemma~\ref{lem:crossed} would have to be adapted.} Let $G = \mathbb{Z}_2 \times \mathbb{Z}_2$, and write $X = (0,1)$, $Z=(1,0)$ and $Y = XZ = (1,1)$ for the non-trivial elements of $G$. Then the map $g \mapsto V_g$, where $V_Y := V_X V_Z$ is a unitary representation of $G$. Define $\alpha_g = \Ad V_g$. The claim is that this defines an action of $G$ on $\mc{R}(\Gamma)$. It is enough to show that $\alpha_g$ maps $\mc{R}(\Gamma)$ into itself, since $V_g^2 = I$ then makes it clear that $\alpha_g$ is an automorphism. Let $\xi$ be a (dual) path in $\Gamma$ and write $F_\xi$ for the corresponding operator. Then $V_g F_\xi = \pm F_\xi V_g$ by equation~\eqref{eq:chcommute}. But these path operators generate a dense subalgebra of $\mc{R}(\Gamma)$ (in the weak topology), and since $\Ad V_g$ is weakly continuous, it follows that $\alpha_g(\mc{R}(\Gamma)) \subset \mc{R}(\Gamma)$. In the next Lemma we show that $\widehat{\mc{R}}(\Gamma)$ is the crossed product of $\mc{R}(\Gamma)$ with respect to this action.
\begin{lemma}\label{lem:crossed}
	The algebra $\widehat{\mc{R}}(\Gamma)$ is isomorphic to a crossed product $\mc{R}(\Gamma) \rtimes_\alpha G$, where $G = \mathbb{Z}_2 \times \mathbb{Z}_2$ and $\alpha$ is the action of $G$ on $\mc{R}(\Gamma)$ defined above. Such an isomorphism $\Phi$ can be chosen in such a way that $\Phi(\mc{R}(\Gamma)) = \mc{R}(\Gamma)$.
\end{lemma}
\begin{proof}
	Identify the Hilbert space $\mc{K} := \mc{H} \otimes \ell^2(G)$ with $\bigoplus_{i=1}^4 \mc{H}$. Since $\alpha$ is implemented by $g \mapsto V_g$, the elements of the crossed product $\mc{R}(\Gamma) \rtimes_\alpha G$ acting on $\mc{K}$ can be identified with matrices~\cite[Ch. 13.1]{MR1468230}
\begin{equation}
	\label{eq:crossedmat}
X = 
\begin{pmatrix}
 R_I & R_X V_X & R_Z V_Z & R_Y V_Y \\
 R_X V_X & R_I & R_Y V_Y & R_Z V_Z \\
 R_Z V_Z & R_Y V_Y & R_I & R_X V_X \\
 R_Y V_Y & R_Z V_Z & R_X V_X & R_I \\
\end{pmatrix},
\end{equation}
with $R_k \in \mc{R}(\Gamma)$. Define a map $\Phi : \mc{R}(\Gamma) \rtimes_\alpha G \to \widehat{\mc{R}}(\Gamma)$ by 
\[
	\Phi( X ) = R_I + R_X V_X + R_Y V_Y + R_Z V_Z,
\]
where $X$ is as above. It is not difficult to check that $\Phi$ is a $*$-homomorphism. Note also that $\Phi(\mc{R}(\Gamma)) = \mc{R}(\Gamma)$, where $\mc{R}(\Gamma)$ (in the argument of $\Phi$) is seen as a subalgebra of $\mc{R}(\Gamma) \rtimes_\alpha G$ in the usual way. The homomorphism $\Phi$ is also normal: let $\rho$ be a normal state on $\widehat{\mc{R}}(\Gamma)$. Then there is an orthogonal sequence of vectors  $\xi_n \in \mc{H}$ with $\sum_n \| \xi_n \|^2 = 1$ and $\rho(A) = \sum_n \langle \xi_n, A \xi_n \rangle$ for all $A \in \widehat{\mc{R}}(\Gamma)$. Define vectors $\Xi_n = \xi_n \oplus 0 \oplus 0 \oplus 0$ and $\Gamma_n = \xi_n \oplus \xi_n \oplus \xi_n \oplus \xi_n$ in $\mc{K}$. Then $X \mapsto \sum_n \langle \Gamma_n, X \Xi_n \rangle_{\mc{K}}$ defines a normal linear functional $\varphi$ on $\mc{R}(\Gamma) \rtimes_\alpha G$. But with $X$ again as in equation~\eqref{eq:crossedmat}, one calculates
\[
\varphi(X) = \sum_n \langle \xi_n, (R_I + R_X V_X + R_Y V_Y + R_Z V_Z) \xi_n \rangle_{\mc{H}} = \rho \circ \Phi(X),
\]
hence $\rho \circ \Phi$ is a normal state on $\mc{R}(\Gamma) \rtimes_\alpha G$ for every normal state $\rho$, so that $\Phi$ is normal as well. This implies that its image $\Phi(\mc{R}(\Gamma) \rtimes_\alpha G)$ is again a von Neumann algebra, and from the definition of $\Phi$ it is clear that this algebra is contained in $\widehat{\mc{R}}(\Gamma)$. But $\widehat{\mc{R}}(\Gamma)$ is the \emph{smallest} von Neumann algebra containing $\alg{A}(\Gamma) \cup \{V_X, V_Y, V_Z\}$, hence the result follows if we can show that $\Phi$ is injective.

To this end, recall that $\Gamma = \Lambda_1 \cup \Lambda_2$. Choose a cone $\Lambda$ such that $\Lambda_1 \ll \Lambda$ and $\Lambda_2 \subset \Lambda^c$, which can always be done by the definition of $\mathcal{C}^2$. Note that by the distal split property, $\mc{R}(\Gamma) \cong \mc{R}(\Lambda_1) \otimes \mc{R}(\Lambda_2)$. In fact, one can explicitly describe a unitary $U$ that decomposes $\mc{H}$ in a tensor product of three Hilbert spaces in a compatible way~\cite{haagdtoric}. In essence, there is a map $U : \mc{H} \to \mc{H}_{\Lambda_1} \otimes \mc{H}_{\Lambda^c} \otimes \mc{H}_0$, with $U \mc{R}(\Lambda_1) U^*$ acting on the first component of the tensor product, and $U \mc{R}(\Lambda_2) U^*$ on the second component. Here $\mc{H}_{\Lambda_1}$ is the closed linear span of $\mc{R}(\Lambda_1) \Omega$, and $\mc{H}_{\Lambda^c}$ is defined similarly. In the following we will write $\mc{K} = \mc{H}_{\Lambda_1} \otimes \mc{H}_{\Lambda^c}$.

As part of the construction of $\mc{H}_0$ (and $U$), one has to choose a path $\widehat{\xi}_X$ (and a dual path $\widehat{\xi}_Z$) from the boundary of $\Lambda_1$ to the boundary of $\Lambda^c$. Without loss of generality, we can assume that these paths coincide with $\xi_X$ or $\xi_Z$ on $\Lambda_1^c \cap \Lambda$. We get four orthonogal vectors in $\mc{H}_0$, namely $\Psi_I := \Omega$, $\Psi_X = F_{\widehat{\xi}_X} \Omega$, $\Psi_Y = F_{\widehat{\xi}_X} F_{\widehat{\xi}_Z} \Omega$ and $\Psi_Z = F_{\widehat{\xi}_Z} \Omega$. Now define projections $P_k = I \otimes |\Psi_k\rangle\langle \Psi_k| \in \alg{B}(\mc{K}) \otimes \alg{B}(\mc{H}_0)$. Note that the projections commute with $U \mc{R}(\Gamma) U^*$.

Now let $X \in \mc{R}(\Gamma) \rtimes_\alpha G$ be represented as in equation~\eqref{eq:crossedmat}. Consider a vector $\xi \in \mc{K}$. We claim that if $\Phi(X) U^* \xi \otimes \Omega = 0$ for all $\xi \in \mc{K}$, then it follows that $R_i = 0$ for $i \in \{I, X,Y,Z\}$. So suppose that $\Phi(X) U^* \xi \otimes \Omega = 0$. First note that
\[
P_k U \Phi(X) U^* = U R_I U^* P_k + U \sum_{l=X,Y,Z} R_l U^* P_k U V_l U^*. 
\]
We claim that $P_k U V_l U^* \xi \otimes \Omega = \delta_{k,l} U V_l U^* \xi \otimes \Omega$. We will show this in the case that $\xi = F \Omega$, where $F$ is the product of a finite number of path operators, each of whose paths lies in $\Lambda_1 \cup \Lambda^c$. Such vectors span a dense subset of $\mc{K}$. Moreover, from the definition of $U$ it follows that $U (\xi \otimes \Omega) = F\Omega$. Note that we have seen before, $V$ either commutes or anti-commutes with $F$. Hence we find
\[
U V_l U^* (F\Omega \otimes \Omega) = \pm U F V_l \Omega = U F F_{\xi_1} F_{\widehat{\xi}_l} F_{\xi_2} \Omega = \pm F F_{\xi_1} F_{\xi_2} \Omega \otimes F_{\widehat{\xi}_l} \Omega,
\]
where we set $\xi_1 = \xi_l \cap \Lambda_1$, $\xi_2 = \xi_l \cap \Lambda^c$, with the convention that $\xi_I = \emptyset$ and $F_\emptyset = I$. It follows that $P_k U V_l U^* \xi \otimes \Omega = \delta_{k,l} U V_l U^* \xi \otimes \Omega$ for all $\xi \in \mc{K}$.

Now suppose that $\Phi(X) = 0$ and hence $P_k \Phi(X) = 0$ for any $k$. By the calculation above it follows that $U R_k U^* (F F_{\xi_1} F_{\xi_2}) \Omega \otimes F_{\widehat{\xi}_k} \Omega = 0$. But $U R_k U^*$ only acts on the first tensor factor, so that it follows that $R_k F F_{\xi_1} F_{\xi_2}) \Omega = 0$ for any $F$ as above. Since $R_k$ is completely determined by its restriction to such vectors (compare~\cite[Lemma 3.5]{haagdtoric} or the proof of Lemma~\ref{lem:restriction}), it follows that $R_k = 0$. Since this holds for all $k$ this shows that $\Phi$ is injective as claimed, completing the proof.
\end{proof}
\begin{remark}
It is known that if $\alg{N} \subset \alg{M}$ is a subfactor with finite index, $\alg{M}$ is a finite-dimensional left $\alg{N}$ module~\cite{MR860811}, see also~\cite[Prop. 5.7]{MR1027496}. The proof gives an explicit realisation of this, by showing that every element of $\widehat{\mc{R}}(\Gamma)$ is of the form $\sum_k A_k V_k$ with $A_k \in \mc{R}(\Gamma)$.
\end{remark}

This explicit description of $\widehat{\mc{R}}(\Gamma)$ makes it possible to calculate the index $[\widehat{\mc{R}}(\Gamma) : \mc{R}(\Gamma)]$, and hence the cone index $\mu_{\pi_0}$. This results in the main theorem of this section: a complete classification of the excitations (as defined in Sect.~\ref{sec:charges}) in the toric code.
\begin{theorem}
\label{thm:toricind}
Consider the toric code on the infinite 2D lattice, and write $\pi_0$ for its (unique) ground state representation. Then the cone index $\mu_{\pi_0}$ is equal to four. Moreover, the model has four distinct irreducible charges, where the ground state representation (i.e., absence of an excitation) is considered to be a charge as well. 
\end{theorem}
\begin{proof}
	By the previous Lemma we have that $\widehat{\mc{R}}(\Gamma) = \mc{R}(\Gamma) \rtimes_\alpha G$. It is well known that the index $[\mc{R}(\Gamma) \rtimes_\alpha G : \mc{R}(\Gamma)]$ is equal to $|G|$, see for example~\cite[Example 3.11]{MR1662525}. In that example it is assumed that $\alpha$ is outer, but this is only used to conclude that $\mc{R}(\Gamma) \rtimes G$ is a factor. This is true in our case since $\widehat{\mc{R}}(\Gamma)$ is a factor. In fact, from the Lemma above and Lemma~\ref{lem:factors}, it also follows that $\alpha$ is outer (by~\cite[Prop. 1.4.4(i)]{MR1473221}).

Note that it follows from the proofs in this section that in fact the choice of cones $\Lambda_1$ and $\Lambda_2$ is not important, which shows that $\mu_{\pi_0} = 4$. In~\cite{toricendo} we identified four distinct (equivalence classes of) charges of the toric code model. By Theorem~\ref{thm:upperbound} the second claim then follows immediately.
\end{proof}
Note that each charge in the toric code is self conjugate, since $\rho^2(A) = A$ for the corresponding localised automorphisms. In particular one finds $d(\rho_i)^2 = 1$ and $\sum_i d(\rho_i)^2 = \mu_{\pi_0}$ as expected from Theorem~\ref{thm:upperbound}.

\section{Non-degeneracy of the sectors}
In the setting above, a braiding can be naturally defined on the class of localised and transportable endomorphisms (c.f.~\cite{MR1405610,MR1838752,toricendo}). Concretely, if $\rho$ and $\sigma$ are two localised and transportable endomorphisms, this construction assigns a unitary operator $\varepsilon_{\rho,\sigma}$ intertwining the endomorphisms $\rho \otimes \sigma$ and $\sigma \otimes \rho$, where $\rho \otimes \sigma := \rho\circ \sigma$.\footnote{More precisely, the endomorphism $\rho$ actually has to be extended to an auxiliary algebra $\alg{A}^{\Lambda_a}$ containing the cone algebras $\mc{R}(\Lambda)$. This is just for technical reasons; the details are not important here and can be found in, e.g., \cite{toricendo}.} The tensor product can be interpreted as first adding a charge of type $\sigma$, and then subsequently a charge of type $\rho$. The assignment $(\rho, \sigma) \mapsto \varepsilon_{\rho,\sigma}$ satisfies all the properties required of a braiding. This has been well-studied in the context of algebraic quantum field theory (see for example~\cite{MR1016869}, where the occurrence of braid, as opposed to symmetric, statistics is discussed). We will be interested in the question how far the braiding is removed from being a symmetry (as is the case in a theory with only Bose/Fermi exchange statistics).

\begin{definition}[\cite{MR1147467}]
	\label{def:degenerate}
A localised and transportable endomorphism $\rho$ is said to be \emph{degenerate} if for \emph{any} irreducible localised and transportable endomorphism $\sigma$ the monodromy is trivial, in the sense that $\varepsilon_{\rho,\sigma} \circ \varepsilon_{\sigma,\rho} = I$.
\end{definition}
This definition only depends on the equivalence class of $\rho$, hence it makes sense to say that a sector is degenerate. In theories where all charges obey symmetric (Bose/Fermi) statistics, the monodromy is always trivial. Hence the set of degenerate sectors tells us something about how far the theory is from being trivial in the sense that there are no anyonic (or braided) statistics. Note that the set of degenerate endomorphisms is always non-empty: the identity morphism (corresponding to the absence of a charge) is always degenerate. The same is true for direct sums of degenerate morphisms. Here a direct sum of two localised (in some cone $\Lambda$) morphisms $\rho$ and $\sigma$ is defined by
\[
	(\rho \oplus \sigma)(A) := V_1 \rho(A) V_1^* + V_2 \sigma(A) V_2^*,
\]
where $V_i \in \mc{R}(\Lambda)$ and $V_i^*V_j = \delta_{i,j} I$ and $V_1 V_1^* + V_2 V_2^* = I$. A localised endomorphism is said to be \emph{trivial} if it is a direct sum of copies of the identity morphism.

Another quantity that contains information about the sectors and the braiding of charges is the modular matrix $S$ (in the sense of Verlinde~\cite{MR954762}). Since we have assumed the existence of conjugate charges and we can define a braiding, the entries of the $S$-matrix can be defined in the present setting as well (see for example~\cite{MR1721563,MR1147467}). A natural question is if the matrix $S$ is invertible. If $S$ is not invertible, the theory cannot be described by a topological quantum field theory (c.f.~\cite[Sect. 4.6]{Wang}). Here we show that a criterion given by M\"uger~\cite{MR1721563} can be generalised to the setting of 2D lattice systems. In the language of category theory, this implies that the category of cone-like localisable and transportable endomorphisms is a \emph{modular tensor category}. Such categories, among other things, give an abstract mathematical description of the properties of the excitations that are used for topological quantum computation~\cite{Wang}.

Physically the absence of non-degenerate excitations can be interpreted as follows. For every non-trivial elementary excitation $\rho$, there is some other excitation $\sigma$ such that interchanging $\rho$ with $\sigma$ twice is a non-trivial operation. Alternatively we can think of $\rho$ being fixed, and $\sigma$ making a full loop around $\sigma$. This implies that the presence or absence of a charge $\rho$ can be detected by pulling a pair of $\sigma$ and its conjugate charge from the vacuum, circle $\sigma$ around some region, and fuse it again with its antiparticle. If a charge $\rho$ is present in the region, this fusion will have a non-trivial result. In this way the presence of any charge can be detected. This can be used in topological quantum computation to get a readout of the result.

It is known by a result of Rehren~\cite{MR1147467} that the two seemingly different concepts of non-degeneracy, namely non-triviality of the braiding and invertibility of $S$, are tightly related. It turns out that (for theories with finitely many different charges) $S$ is invertible if and only if all degenerate morphisms are trivial. It is precisely this fact that we will exploit to give a sufficient condition that implies that $S$ is invertible. The following lemma is a crucial part. Its proof is an adaptation of the proof of~\cite[Prop. 4.2]{MR1721563} for algebraic quantum field theories on the line, to the present situation of 2D quantum spin systems. For the convenience of the reader we present the details of the proof.
\begin{lemma}\label{lem:inner}
	Let $\rho$ be an irreducible endomorphism localised in a cone $\Lambda_1$. Suppose that $\Lambda_2$ is a cone such that $\Lambda_1 \ll \Lambda_2$ be two cones and the opening angle of $\Lambda_2$ is bigger than that of $\Lambda_1$. If for every such cone $\Lambda_2$, $\rho$ acts trivially on $\mc{R}(\Lambda_1)' \cap \mc{R}(\Lambda_2)$, then $\rho$ is inner, in the sense that it is a direct sum of copies of the trivial representation.
\end{lemma}
\begin{proof}
Choose a cone $\widetilde{\Lambda}$ such that $\Lambda_2 \ll \widetilde{\Lambda}$. Then by the distal split property there are Type I factors $\alg{N}_1$ and $\alg{N}_2$ such that
\[
\mc{R}(\Lambda_1) \subset \alg{N}_1 \subset \mc{R}(\Lambda_2) \subset \alg{N}_2 \subset \mc{R}(\widetilde{\Lambda}).
\]
Let $A \in \alg{N}_1$ and $B \in \alg{N}_1' \cap \mc{R}(\widetilde{\Lambda})$. Then by the assumptions of the Lemma, it follows that $\rho(B) = B$ and by localisation and Haag duality that $\rho(\alg{N}_1) \subset \mc{R}(\Lambda_2)$, hence
\[
\rho(A) \in (\alg{N}_1' \cap \mc{R}(\widetilde{\Lambda}))' \cap \alg{N}_2 \subset (\alg{N}_1' \cap \alg{N}_2)' \cap \alg{N}_2 = \alg{N}_1,
\]
since the $\alg{N}_i$ are Type I factors. Hence $\rho$ restricts to an endomorphism of $\alg{N}_1$ and therefore there is an (at most countable) collection of isometries $V_i \in \alg{N}_1$ (by~\cite[Cor. 3.8]{MR872351}) with pairwise orthogonal ranges, such that $\sum_i V_i V_i^* = I$ and
\[
	\rho(A) = \eta(A) := \sum_i V_i A V_i^*
\]
for all $A \in \alg{N}_1$. By assumption $\eta(A) = A$ for each $A \in \mc{R}(\Lambda_1)' \cap \mc{R}(\Lambda_2)$. By multiplying on the right by $V_i$ for some $i$ it follows that
\[
V_i \in (\mc{R}(\Lambda_1)' \cap \mc{R}(\Lambda_2))' \cap \alg{N}_1 = \mc{R}(\Lambda_1).
\]
It is clear that $\eta(A)$ can be defined also for all $A \in \mc{R}(\Lambda_1)' = \mc{R}(\Lambda^c)$, and it then follows that $\eta(A) = A = \rho(A)$ for all $A \in \alg{A}(\Lambda^c_1)$. Since they also agree on $\alg{A}(\Lambda_1)$ by construction of $\eta$, it follows that $\eta = \rho$ on all of $\alg{A}$.
\end{proof}
Note that in the case of spin systems it is easier to conclude that $\rho = \eta$ on all observables, essentially because we have to deal with \emph{complements} of regions, instead of \emph{spacelike} complements. In~\cite{MR667768} the time slice axiom was used to make a similar conclusion. In the present situation we can view the plane as the $t = 0$ ``time slice'', and this already completely determines the system.

\begin{figure}
	\begin{center}
	\includegraphics[width=7cm]{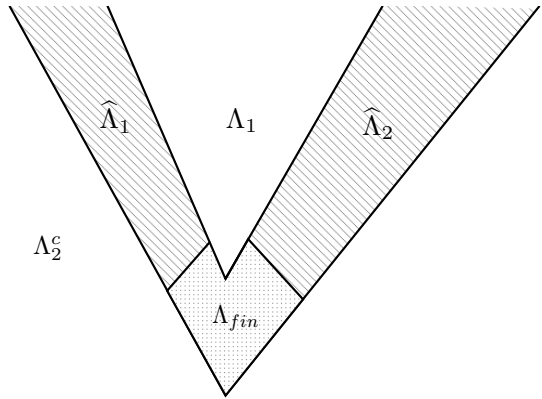}
\end{center}
\caption{Two cones $\Lambda_1 \ll \Lambda_2$. The region $\Lambda_1^c \cap \Lambda_2$ is subdivided into two truncated cones $\widehat{\Lambda}_i$ (dashed region) and a remaining part $\Lambda_{fin}$ (dotted region) consisting of only finitely many sites. The finite region is chosen to be as big as necessary to apply the split property to the truncated cones in the shaded regions, i.e., $\widehat{\Lambda}_1 \cup \widehat{\Lambda}_2 \in \mc{C}^2_{tr}$.}
	\label{fig:truncated}
\end{figure}

In the proof the condition that the opening angle of $\Lambda_2$ is bigger than that of $\Lambda_1$ is not used. It will be important in the following, however. As remarked before, the properties of Haag duality and the distal split property do not crucially depend on the cone shape. We will henceforth assume that they also hold for \emph{truncated cones} as in Figure~\ref{fig:truncated}. An inspection of the proofs in~\cite{haagdtoric} shows that this is the case for the toric code. The set of pairs of truncated cones $\Lambda_1 \cup \Lambda_2$ such that the split property holds for $\mc{R}(\Lambda_1) \vee \mc{R}(\Lambda_2)$ will be denoted by $\mc{C}^2_{tr}$.
\begin{theorem}
	Suppose that for all cones $\Lambda_1$ and $\Lambda_2$ such that $\Lambda_1 \ll \Lambda_2$ and $\Lambda_1^c \cap \Lambda_2$ can be written as the union of some $\Gamma = \widehat{\Lambda}_1 \cup \widehat{\Lambda}_2 \in \mathcal{C}^2_{tr}$ and a finite region $\Lambda_{fin}$ (see Figure~\ref{fig:truncated}), the algebra $\mc{R}(\Gamma^c)'$ is generated by $\mc{R}(\widehat{\Lambda}_1)$, $\mc{R}(\widehat{\Lambda}_2)$ and the set $\mc{T}$ of charge transporters from $\widehat{\Lambda}_1$ to $\widehat{\Lambda}_2$ (and vice versa), then each non-trivial sector is non-degenerate. If there are only finitely many (equivalence classes of) charges, the modular matrix $S$ is invertible.
\end{theorem}
\begin{proof}
	Let $\rho$ be a degenerate irreducible endomorphism localised in some cone $\Lambda_1$. Choose $\Lambda_2$ such that $\Lambda_1 \ll \Lambda_2$ be as in the premises. Then by~\cite[Lemma 3.2]{MR1721563}, the proof of which carries over almost verbatim to the present situation, $\rho$ is degenerate if and only if $\rho(V) = V$ for each charge transporter $V$ from $\widehat{\Lambda}_1$ to $\widehat{\Lambda}_2$ (and vice versa). By localisation, and since $\rho$ is normal when restricted to $\mc{R}(\Lambda)$ for any cone $\Lambda$, it follows that $\rho$ acts trivially on $\rho(\widehat{\Lambda}_1)$ and $\rho(\widehat{\Lambda}_2)$. By localisation it also acts trivially on $\alg{A}(\Lambda_{fin})$, and by the degeneracy assumption on the charge transporters $\mc{T}$ as well. It follows that $\rho$ acts trivially on the algebra 
\[
\mc{R}(\widehat{\Lambda}_1) \vee \mc{R}(\widehat{\Lambda}_2) \vee \alg{A}(\Lambda_{fin}) \vee \mc{T}.
\]
By assumption this is equal to $\mc{R}(\Gamma^c)' \vee \alg{A}(\Lambda_{fin})$.

To proceed, note that $\mc{R}(\Gamma^c) \simeq \alg{A}(\Lambda_{fin}) \otimes \mc{R}(\Gamma^c \setminus \Lambda_{fin})$, since $\alg{A}(\Lambda_{fin})$ is a finite Type I factor. By the same reason the Hilbert space $\mc{H}$ on which the von Neumann algebras act can be decomposed as $\mc{H}_1 \otimes \mc{H}_2$, where $\alg{A}(\Lambda_{fin}) \simeq I \otimes \alg{B}(\mc{H}_2)$. It follows that
\[
\mc{R}(\Gamma^c) \cap \alg{A}(\Lambda_{fin})' \simeq \mc{R}(\Gamma^c \setminus \Lambda_{fin}) \otimes \alg{A}(\Lambda_{fin}) \cap \alg{B}(\mc{H}_1) \otimes I = \mc{R}(\Gamma^c \setminus \Lambda_{fin}),
\]
and therefore $\mc{R}(\Gamma^c)' \vee \alg{A}(\Lambda_{fin}) = \mc{R}(\Gamma^c \setminus \Lambda_{fin})'$. Note that the region $\Gamma^c \setminus \Lambda_{fin} = \Lambda_1 \cup \Lambda_2^c$. By the distal split property it follows that $\mc{R}(\Gamma^c \setminus \Lambda_{fin}) \simeq \mc{R}(\Lambda_1) \vee \mc{R}(\Lambda_2^c)$. Hence by Haag duality
\[
\mc{R}(\Gamma^c \setminus \Lambda_{fin})' = (\mc{R}(\Lambda_1) \vee \mc{R}(\Lambda_2)')' = \mc{R}(\Lambda_1)' \cap \mc{R}(\Lambda_2).
\]
By Lemma~\ref{lem:inner} it then follows that each degenerate irreducible endomorphism is trivial. The last claim then follows by a result of Rehren~\cite{MR1147467}, who proved that $S$ is invertible if and only if the only degenerate endomorphisms are trivial.
\end{proof}
This gives an alternative proof that in the case of the toric code the corresponding category of localised endomorphisms is modular, since the assumptions of the theorem above are satisfied because of Lemma~\ref{lem:vnagen}. Of course that already follows since by Theorem~\ref{thm:toricind} of the present work and Theorem 6.2 of Ref.~\cite{toricendo}, together with the well-known fact that the category of representations of the quantum double of a finite group is a modular tensor category (see e.g.~\cite{MR1797619}, or \cite[Ch. 5]{phdnaaijkens} for a proof by showing that the centre of the category, that is the full subcategory of degenerate objects in the sense of Definition~\ref{def:degenerate}, is trivial). However, this criterion may be of use in models for which the full category of endomorphisms is not known. 

\vspace{\baselineskip}
\textbf{Acknowledgements:} The author is supported by the Dutch Organisation for Scientific Research (NWO) through a Rubicon grant and partly through the EU project QFTCMPS and the cluster of excellence EXC 201 Quantum Engineering and Space-Time Research. The index approach was suggested to the author independently by K.-H. Rehren and M. M\"uger. The author wishes to thank Tobias Osborne for many helpful comments and conversations.

\bibliographystyle{abbrv}
\bibliography{../../bib/refs}

\end{document}